\documentclass[12pt]{article}

\makeatletter
\def\@fnsymbol#1{\ensuremath{\ifcase#1\or *\or
		\mathsection\or \mathparagraph\or \|\or **\or \dagger\dagger
		\or \ddagger\ddagger \else\@ctrerr\fi}}
\usepackage{color,amsmath,amsfonts,graphicx} 
\usepackage{amssymb}
\usepackage{fancyhdr}
\usepackage{latexsym}
\usepackage{amsthm}

\newtheorem{theorem}{Theorem}
\newtheorem{corollary}{Corollary}
\newtheorem{definition}{Definition}
\newtheorem{example}{Example}

\newtheorem{proposition}{Proposition}
\newtheorem{remark}{Remark}

\newtheorem{property}{Property}

\renewcommand{\epsilon}{\varepsilon}

\renewcommand{\hat}{\widehat}

\newcommand{\ran}{\mbox{\boldmath $l$}}
\newcommand{\ranc}{\mbox{\boldmath $c$}}

\newcommand{\mwcn}{\mathcal{W}^N_{\min}}

\newcommand{\wcv}{\mathcal{W}^v}
\newcommand{\mwcv}{\mathcal{W}_{\min}^v}
\newcommand{\mwcvv}{\mathcal{W}^{v^*}_{\min}}
\newcommand{\mwcvvp}{\mathcal{W}^{v^*_\pi}_{\min}}
\newcommand{\mwcvp}{\mathcal{W}^{v_\pi}_{\min}}
\newcommand{\mwcpvv}{\mathcal{W}^{v'^{*}}_{\min}}
\newcommand{\mwcpv}{\mathcal{W}^{v'}_{\min}}

\newcommand{\bcv}{\mathcal{B}^v}

\newcommand{\mbcv}{\mathcal{B}_{\min}^v}

\newcommand{\Min}{\mbox{Min }}

\usepackage{xcolor}

\date{ }
\begin{document}
	
	\title{Lexicographic Ranking based on Minimal Winning Coalitions
	}
	

	%
	\author{M. Aleandri\footnote{Luiss University, Viale Romania, 32, 00197 Rome, Italy. \textit{maleandri@luiss.it}   },
		V. Fragnelli\footnote{University of Eastern Piedmont, Department of Sciences and Innovative Technologies (DISIT), Viale T. Michel 11, 15121 Alessandria, Italy. \textit{franco.fragnelli@gmail.com}},
		S. Moretti\footnote{LAMSADE, CNRS, Universit\'e Paris-Dauphine, Universit\'e PSL, 75016 Paris, France. \textit{stefano.moretti@dauphine.fr} }           
	}
	%
	
	\maketitle
	
	\begin{abstract}
		In this paper, we consider the consistency of the desirability relation with the ranking of the players in a simple game provided by some well-known solutions, in particular the Public Good Index \cite{H82} and the criticality-based ranking \cite{ADFM21}. 
We define a new ranking solution, the Lexicographic Ranking based on Minimal winning coalitions (LRM), strongly related to the Public Good Index being rooted in the minimal winning coalitions of the simple game, proving that it is monotonic with respect to the desirability relation \cite{Isbell58}, when it holds. A suitable characterization of the LRM solution is provided.
Finally, we investigate the relation among the LRM solution and the criticality-based ranking, referring to the dual game.\\


\noindent
\textbf{keywords}: Desirability, Simple Games, Public Good Index,  Ranking, Dual Game
	\end{abstract}

\section{Introduction}
\label{intro}

The \textit{desirability relation} for players in simple games \cite{Carreras95, Freixas12,Isbell58,Zwicker99}  has been widely studied  also in connection with the property-driven analysis of power indices  \cite{Diffo02, Freixas97, Holler13}. A player $i$ is in desirability relation with a player $j$ (meaning that $i$ is at least as desirable as $j$) if we can replace player $j$ with player $i$ in any winning coalition  without changing the outcome, i.e. for any  winning coalition $S$ with $j \in S$ and $i \notin S$, we have that  $S \setminus \{j\} \cup \{i\}$ is still a winning coalition. So, the desirability relation between $i$ and $j$ suggests that player $i$ is at least as influential as  
player $j$, for it is never harmful for coalitions to replace $j$ by $i$. Clearly, the desirability relation is not necessarily a total relation on the set of players, as two players may not be in any desirability relation (see, for example, \cite{Freixas05, Freixas08, Holler13} for an analysis of properties of the desirability relation on simple games).  When the desirability relation is a total preorder the simple game is called complete and, for example, weighted games are complete, see \cite{Freixas10}.

A classical property for power indices based on the desirability relation, is the monotonicity property:  a power index  is monotone if, whenever a player $i$ is at least as desirable as player $j$, then the power of $i$ is at least as much as the power of player $j$ (see Remark \ref{Remark:mono}). For instance, the Shapley-Shubik index \cite{Shapley54}, the Banzhaf index \cite{Banzhaf65}, the Johnston index \cite{Johnston78}, the nucleolus \cite{Schmeidler69} and many other power indices satisfy the monotonicity property (see, for instance, \cite{Freixas97}). Instead, it is easy to provide examples of simple games showing this property is not satisfied by other famous indices that take into account exclusively minimal winning coalitions, like the Deegan-Packel index (DPI) \cite{Deegan80} or the Public Good Index (PGI) \cite{H82, Holler83} (see Section \ref{sec:pgr} for some examples from the literature  \cite{ DeeganPackel80, Freixas97}).

The main objective of this paper is to show that it is possible to rank players  consistently with the desirability relation and using exclusively minimal winning coalitions in an ordinal way. In fact, while classical power indices convert the  information about coalitions  into a numerical personal score representing players' relevance in a simple game, in many practical situations, having a reliable ranking to select the top players is enough and the information provided by players' score is only marginal. For instance, in the application of power indices to computational biology,  the goal is short-listing the most relevant genes on complex networks with a huge number of nodes \cite{Moretti07}. In a similar way, ranking players is essential for the analysis of centrality of  network elements with the goal to select the most critical or sensible parts of a system \cite{Lindelauf13}, or in studies aimed at establishing which agents are the strongest or the weakest in a voting system \cite{Ferto20}.

To that purpose, we introduce a \textit{ranking solution} (formally, a map that associates to any simple game with player-set $N$, a total preorder on $N$) aimed at ranking players in a simple game according to their influence and in a way that is compatible with the desirability relation. Our ranking solution  contains elements of both the DPI and the PGI,  taking into account  the minimal winning coalitions  an individual belongs to.  More exactly, given $n=|N|$ players, we first compute for each player a vector of $n$ real numbers, where the $k$-th component of each vector is the number of minimal winning coalitions of size $k$, with $k=1, \ldots, n$; second,  our ranking solution lexicographically compares  those real-valued vectors. Due to the similarity with the PGI to compute vectors components, we called such a ranking solution the \textit{Lexicographic Ranking based on Minimal winning coalitions} (shortly, the \textit{LRM}). 

We show that the \textit{LRM} is monotonic, and we prove that it is the unique solution satisfying (strong) monotonicity with respect to the desirability relation together with  two other axioms: 1) the \textit{coalitional anonymity} property, saying that the relative ranking between two players $i$ and  $j$ in two different simple games should
be independent of the identity of other players in minimal winning coalitions, provided that the number of minimal winning coalitions to which they belong in the two games is the same; 2) the property of \textit{independence of larger minimal winning coalitions}, saying that once a player $i$ is considered more influential
than a player $j$ in a simple game, player $i$ will continue to be considered more influential than $j$ in any simple game obtained by adding new ``larger'' minimal winning coalitions to the original game.

As a side-product of our analysis, we also point out some connections between our ranking solution and the criticality-based ranking provided  in \cite{ADFM21} to compare, in an ordinal way, the blocking power of players and we explore some similarities with the axioms  used to characterize the dual version of the criticality-based ranking.

The paper is structured as follows.
In Section \ref{pre} we provide some basic notions and notation. The definition of the \textit{LRM} is then introduced in Section \ref{sec:pgr} together with some examples comparing it with the ranking defined by other solutions from the literature. An axiomatic characterization of the \textit{LRM} is then presented and discussed in Section \ref{sec:ax}. A connection between the criticality-based ranking and the \textit{LRM} are investigated in Section \ref{sec:dual} using the desirability relation on dual games. Section \ref{sec:concl} concludes.

\section{Preliminaries and notation} \label{pre}

Given a finite set $N$, we denote by $|N|$ its \textit{cardinality} and by $2^N=\{S \subseteq N\}$ its \textit{power set}. 
A \emph{simple game} is a pair $(N,v)$, where $N=\{1,2,\ldots,n\}$ denotes a finite set of players and $v:2^N\to\{0,1\}$ is a \emph{characteristic function}, with $v(\varnothing)=0$, $v(N)=1$ and $v(S)\leq v(T)$ for all sets $S,T$ such that $S\subseteq T\subseteq N$. 
A coalition $S \subseteq N$ such that $v(S)=0$  is said a \emph{losing} coalition, whereas a coalition $S \subseteq N$ such that $v(S)=1$ is said a \emph{winning} coalition.  The class of simple games with $N$ as the set of players is denoted by $\mathcal{SG}^N$.\\
%
%
Let $\wcv$ be the set of winning coalitions in $(N,v)$
\[
\wcv= \{ S \subseteq N: v(S)=1 \}
\]
and let $\mwcv$ be the set of minimal winning coalitions in $(N,v)$
\[
\mwcv = \Min \wcv
\]
where, for any family of sets $\mathcal{F}$, the $\Min$ operator on $\mathcal{F}$ removes all non-inclusion-minimal sets of $\mathcal{F}$:
\[
\Min \mathcal{F}= \left\{ F \in \mathcal{F}| \nexists G \in \mathcal{F} : G \subset F \right\}.
\]
A simple game $(N,v)$ is a \textit{weighted majority game} if there exists
a vector of non-negative real numbers $w \in \mathbb{R}^N_{\geq 0}$ and a quota $q \in \mathbb{R}_{\geq 0}$ such that a coalition $S\subseteq N$ is winning
if and only if $\sum_{i \in S} w_i \geq q$.

In \cite{H82} the author introduced the \textit{Public Good index} (\textit{PGI}) of a player in a simple game, as the quotient between the number of minimal winning coalitions containing that player and the sum of cardinalities of all the minimal winning coalitions. Let $(N,v)$ be a simple game, the PGI of player $i\in N$:
\begin{equation*}
	h_v(i)=\frac{|\mwcv(i)|}{\sum_{j\in N}|\mwcv(j)|}
\end{equation*}
where $\mwcv(i)=\{W\in\mwcv:\, i\in W\}$.

\begin{example}\label{ex: 1} 
	Let $N=\{1,2,3,4,5\}$, and let 
	\[
	\mwcv =\{\{1,2\}, \{1,3\}, \{3,4\}, \{2,4,5\}, \{1,4,5\} \}.
	\]
	We have 
	\begin{equation*}
	h_v(1)=\frac{3}{12},\ h_v(2)=\frac{2}{12},\ h_v(3)=\frac{2}{12},\ h_v(4)=\frac{3}{12},\ h_v(5)=\frac{2}{12}.    
	\end{equation*}
\end{example}
In Example \ref{ex: 1} we can observe that player $3$ and $5$ have the same PGI, but they belong to minimal winning coalitions of different cardinality and whenever player $5$ is winning together with some coalition $S\subseteq N\setminus\{3,5\}$ then player $3$ is winning together with the same coalition. 

In \cite{Deegan80} the authors measure the power of a player according to the size of the minimal winning coalitions she belongs to. So, the \textit{Deegan-Packel index} (\textit{DPI}) for player $i$ is defined as:
\begin{equation*}
	\delta_v(i)=\sum_{W\in\mwcv(i)}\frac{1}{|\mwcv|}\frac{1}{|W|}.
\end{equation*}
\begin{example}Taking the same simple game $(N,v)$ of Example \ref{ex: 1} we have: 
	\begin{equation*}
		\delta_v(1)=\frac{8}{30},\ \delta_v(2)=\frac{5}{30},\ \delta_v(3)=\frac{6}{30},\ \delta_v(4)=\frac{7}{30},\ \delta_v(5)=\frac{4}{30}.    
	\end{equation*}
\end{example}
\noindent
According to the DPI, player $3$, for instance, has more power than player $5$ because it belongs to two minimal winning coalitions of size smaller than the two minimal winning coalitions containing player $5$. 

\section{A ranking solution and the desirability relation}\label{sec:pgr}

 Let us start recalling that a binary relation on $N$ is a subset of $N \times N$.  A reflexive, transitive and total binary relation on $N$ is a \textit{total preorder} (also called, a \textit{ranking}) on $N$. We denote by $\mathcal{T}^N$ the set of all total preorders on $N$. For instance, consider the lexicographic total preorder among vectors of real numbers:
 $$\mathbf{x} \ge_L \mathbf{y} \quad  \mbox{ if either } \mathbf{x}=\mathbf{y} \;\;
 \mbox{ or }\;  \exists k: x_t=y_t, t=1,\dots, k-1\; \mbox{ and } x_k>y_k.$$

  We define a \textit{ranking solution} or, simply, a \textit{solution}, as a map $R: \mathcal{SG}^N \rightarrow \mathcal{T}^N$  that associates to each simple game $v \in \mathcal{SG}^N$ a total preorder on $N$. The value assumed by a map $R$ on a simple game $v$ is the ranking on $N$  denoted by $R^v$. We use the notation $i R^v j$ to say that $(i,j) \in R^v$, and it means that $i$ is at least as important as $j$ according to ranking $R^v$,  for all $i,j \in N$. We denote by $I^v$ the symmetric part of $R^v$, i.e. $i I^v j$ means that $(i,j)\in R^v$ and $(j,i)\in R^v$ ($i$ and $j$ are equivalent), and by $P^v$ its asymmetric part, i.e. $i P^v j$ means that $(i,j)\in R^v$ and $(j,i)\notin R^v$ ($i$ is strictly more important than $j$).

Clearly, any real-valued $N$-vector numerically represents a total preorder over the player set $N$. Consequently, any power index $\phi: \mathcal{SG}^N \rightarrow \mathbb{R}^N$ underpins a ranking solution denoted by $R_{\phi}$ and such that $i R^v_{\phi} j \Leftrightarrow \phi_i(v) \geq \phi_j(v)$.

In this section, we introduce a new ranking solution for simple games based on minimal winning coalitions. 
The main idea of the new solution is that the smaller is the size of a minimal winning coalition, the larger is the power of its members. Therefore, the ranking of a player  is positively correlated first to the size of minimal winning coalitions the player belongs to and, second, to their number. 


To define the ranking solution, we need to introduce the notation $i_k$ representing the number of minimal winning coalitions of size $k$  containing  $i$  in a simple game $(N,v)$: $i_k=|\{S \in \mwcv: i \in S, |S|=k\}|$ for all $k \in \{1, \ldots, n\}$. For each $i \in N$, let $\theta_v(i)$ be the $n$-dimensional vector $\theta_v(i)=(i_1,\dots,i_n)$ associated to $v$.

\begin{definition}\label{def:LRM}[Lexicographic Ranking based on Minimal winning coalitions (LRM)]
	The \emph{Lexicographic Ranking based on Minimal winning coalitions (LRM)} solution is the function $R_{\ran}:\mathcal{SG}^N \longrightarrow \mathcal{T}^N$ defined for any  simple game $v \in \mathcal{SG}^N$ as
	$$
	i \  R_{\ran}^v \ j \qquad {\rm if\;} \qquad \theta_v(i)\;\;\ge_L \;\;\theta_v(j).
	$$
	Let $I_{\ran}^v$ and $P_{\ran}^v$ be the symmetric part and the asymmetric part of $R_{\ran}^v$, respectively.
\end{definition}


\begin{example}\label{ex:LRM1}
Consider the simple game of Example \ref{ex: 1}. We have that 
\begin{eqnarray*}
&\theta_{v}(1)=(0, 2, 1, 0, 0),\ &\theta_{v}(2)=(0, 1, 1, 0, 0),\\
&\theta_{v}(3)=(0, 2, 0, 0, 0),\ &\theta_{v}(4)=(0, 1, 2, 0, 0),\\
&\theta_{v}(5)=(0, 0, 2, 0, 0).
\end{eqnarray*}
So, the LRM solution $R_{\ran}^v$ ranks the players as follows
\begin{equation*}
	1\ P_{\ran}^v\ 3\ P_{\ran}^v\ 4\ P_{\ran}^v\ 2\ P_{\ran}^v\ 5.
\end{equation*}
Notice that the ranking provided by the PGI and the DPI do not coincide with the ranking $R_{\ran}^v$ on this example. In fact, for instance, $h_v(4)>h_v(3)$ and $\delta_v(4)>\delta_v(3)$, while $3\ P_{\ran}^v\ 4$.
\end{example}

The LRM solution always provides a total preorder over the player set $N$ for any simple game $(N,v)$. Instead, given a simple game $(N,v)$, the desirability relation \cite{Isbell58} is a preorder over the elements of $N$ and is defined as follows.
\begin{definition}
Let $(N,v)$ be a simple game. For any pair of players $i,j \in N$, the \textit{desirability relation}
$\succeq^v \subseteq N \times N$ is defined as follows:
\[
i \succeq^v j \Leftrightarrow [ S \cup \{j\} \in \mathcal{W}^{v} \Rightarrow S \cup \{i\} \in  \mathcal{W}^{v} \mbox{ for all } S \subseteq N \setminus \{i,j\}].
\]
\end{definition}
In the following, if the game $v$ on which $\succeq^v$ is defined is clear from the context, we denote relation  $\succeq^v$ simply by $\succeq$.
For any $i,j \in N$, $i \succeq j$ is interpreted as player $i$ is at least as desirable as
player $j$ (as a coalitional member); $i \succ j$ means that $i \succeq j$ and there exists a 
 coalition $T \subseteq N \setminus \{i,j\}$ such that $T \cup \{i\} \in \mathcal{W}^{v}$ but $T \cup \{j\} \notin \mathcal{W}^{v}$, and it is interpreted as player $i$ is (strictly) more desirable than player $j$; $i \sim j$ means that $i \succeq j$ and $j \succeq i$, i.e. it is true that $S \cup \{j\} \in \mathcal{W}^{v} \Leftrightarrow S \cup \{i\} \in \mathcal{W}^{v}$ for all $S \subseteq N \setminus \{i,j\}$, and it is interpreted as players $i$ and $j$ are equally desirable.

As discussed in Section \ref{intro}, the desirability relation, when it holds, represents a criterion to select between two players the most influential one, that is the player winning the maximum number of times.   
So, it is interesting to require the following basic property for ranking solutions.
\begin{property}[Desirable Monotonicity (DM)]\label{PD} \rm
Let $i,j \in N$. For any  $v \in \mathcal{SG}^N$, a solution $R$ satisfies the  \textit{desirable monotonicity}  property if
$$
i \sim^{v} j \Rightarrow i I^{v} j,
$$
and
$$
i \succ^{v} j \Rightarrow i P^{v} j.
$$
\end{property}
A solution satisfying the desirable monotonicity should strictly obey to the desirability relation: if the desirability relation between two players  is strict (i.e., $i \succ^v j$), then a ranking solution should put such players in a strict relation too (i.e., $i P^{v} j$); of course, if two players are equally desirable (i.e., $i \sim^{v} j$) then the ranking solution must define the same kind of relation (i.e., $i I^{v} j$). Notice that this kind of ``strong'' monotonicity relation is not satisfied by the ranking over players represented by the nucleolus \cite{Schmeidler69}, as it is easy to find examples of simple games having players in the symmetric part of the desirable relation and such that the allocation provided by the nucleolus is different (see, for instance, \cite{Freixas97} page 600).

The rankings over players represented by the DPI and the PGI do not satisfy the desirable monotonicity property, as shown by the following example.

\begin{example}\label{ex:nomon}
Consider a weighted majority game $(N,v)$, $N=\{1,2,3,4,5\}$, with weight function $(4,2,1,1,1)$ and quota $q=6$. So, the minimal winning coalitions are
\[
\mwcv =\{\{1,2\}, \{1,3,4\}, \{1,3,5\}, \{1,4,5\} \}.
\]
We have that $1 \succ^v 2 \succ^v 3 \sim^v 4 \sim^v 5$. However, according to the PGI $h_2(v)=\frac{1}{11}<\frac{2}{11}=h_3(v)$, while according to the DPI we have $\delta_2(v)=\frac{1}{8}<\frac{1}{6}=\delta_3(v)$. So, according to the rankings underpinned by both indices, player $3$ is ranked strictly higher than player $2$.

On the other hand,
\begin{eqnarray*}
&\theta_{v}(1)=(0, 1, 3, 0, 0),\ &\theta_{v}(2)=(0, 1, 0, 0, 0),\\
&\theta_{v}(3)=(0, 0, 2, 0, 0),\ &\theta_{v}(4)=(0, 0, 2, 0, 0),\\
&\theta_{v}(5)=(0, 0, 2, 0, 0).
\end{eqnarray*}
So, $1\ P_{\ran}^v\ 2\ P_{\ran}^v\ 3\ I_{\ran}^v\ 4\ I_{\ran}^v\ 5$:  $R_{\ran}^v$ and $\succeq^v$ coincide.

\end{example}

In general, a total preorder provided by the LRM solution coincides with the desirability relation on any simple game where the desirability relation is total. This fact is an immediate consequence of the following proposition.

\begin{proposition}\label{prop:pd}
The LRM solution $R_{\ran}$ fulfils the desirable monotonicity property. 
\end{proposition}
\begin{proof}
Let $(N,v)$ be a simple game.
It is easy to verify that the condition
\begin{equation}\label{equivw}
S \cup \{j\} \in \mathcal{W}^{v} \Leftrightarrow S \cup \{i\} \in \mathcal{W}^{v} \mbox{ for all } S \subseteq N \setminus \{i,j\}
\end{equation}
is equivalent to the condition
\begin{equation}\label{equivwmin}
S \cup \{j\} \in \mathcal{W}^{v}_{\min} \Leftrightarrow S \cup \{i\} \in \mathcal{W}^{v}_{\min} \mbox{ for all } S \subseteq N \setminus \{i,j\}.
\end{equation}
We prove that $i \sim^{v} j \Rightarrow i I_{\ran}^{v} j$.\\
Since $i \sim^{v} j $, according to the equivalence between relations (\ref{equivw}) and (\ref{equivwmin}), we immediately have
that $S \cup \{j\} \in \mathcal{W}^{v}_{\min} \Leftrightarrow S \cup \{i\} \in \mathcal{W}^{v}_{\min}$ for all $S \subseteq N \setminus \{i,j\}$. So, $\theta_v(i)=\theta_v(j)$, and therefore $i I_{\ran}^{v} j$.\\ \\
Now, we prove that $i \succ^{v} j \Rightarrow i P_{\ran}^{v} j$.\\

Let be $i \succ^{v} j$ and define $\mathcal{T}=\{ T \subseteq N \setminus \{i,j\}: T\cup \{i\} \in \mathcal{W}^{v}_{\min}, T \cup \{j\} \notin \mathcal{W}^{v}_{\min}\}$. \\\noindent
We first need to prove that  $\mathcal{T}\neq \emptyset$. Since $i \succ^{v} j$, it must exist $ T \subseteq N \setminus \{i,j\}$ such that $ T\cup \{i\} \in \mathcal{W}^{v}$ and $T \cup \{j\} \notin \mathcal{W}^{v}$ and, by the equivalence between relation (\ref{equivw}) and (\ref{equivwmin}), it is not possible that $S \cup \{j\} \in \mathcal{W}^{v}_{\min} \Leftrightarrow S \cup \{i\} \in \mathcal{W}^{v}_{\min} \mbox{ for all } S \subseteq N \setminus \{i,j\}$. Moreover, again for $i \succ^{v} j$, it  is not possible that there exists $S \subseteq N \setminus \{i,j\}$ such that $S\cup \{i\} \notin \mathcal{W}^{v}_{\min}$ and  $S \cup \{j\} \in \mathcal{W}^{v}_{\min}$. So, it must exist $T \subseteq N \setminus \{i,j\}$ such that
 $T\cup \{i\} \in \mathcal{W}^{v}_{\min}$ and $T \cup \{j\} \notin \mathcal{W}^{v}_{\min}$.

Now, let $k=\min\{|T|: T\in\mathcal{T}\}$.
If $k=0$, we immediately have that $\{i\} \in \mathcal{W}^{v}_{\min}$ and $\{j\} \notin \mathcal{W}^{v}_{\min}$, so $i P_{\ran}^{v} j$.\\
Consider the case $k>0$. By the minimality of $k$ we have that $i_t=j_t$ for all $t=0,\ldots,k-1$ and $i_k > j_k$ and so $i P_{\ran}^{v} j$.

\end{proof}

\begin{remark}\label{Remark:mono}
It is well known from the literature that the desirability relation on weighted majority games is a total preorder  and that the following monotonicity condition w.r.t. weights holds for  a weighted majority game $(N,v)$ with weights $(w_1, \ldots, w_n)$:
\[
w_i \geq w_j \Rightarrow i \succeq^v j,
\] 
for all $i,j \in N$ (see for instance \cite{Freixas97}). As a direct consequence of Proposition \ref{prop:pd} we have that also the LRM solution on weighted majority games is monotonic w.r.t. weights, that is $w_i \geq w_j \Rightarrow i R_{\ran}^v j$ for all $i,j \in N$.

\end{remark}

\section{An axiomatic characterization of the LRM solution}\label{sec:ax}

Now, we introduce two new properties for ranking solutions that are inspired by similar properties introduced in \cite{ADFM21}  on the sets of blocking coalitions.

The next property says that winning coalitions of the same size should have the same impact on the ranking, independently of their members.
\begin{property}[Anonymity of Minimal Winning Coalitions (AMWC)]\label{AMWC} \rm
	Let $i,j \in N$,  $v, v_{\pi} \in \mathcal{SG}^N$ and let $\pi$ be a bijection on $2^{N \setminus\{i,j\}}$ with $|\pi(S)|=|S|$ and such that
	$$ S \cup \{i\} \in \mwcv \Leftrightarrow  S \cup \{i\} \in \mwcvp$$
	and
	$$ S \cup \{j\} \in  \mwcv \Leftrightarrow  \pi(S) \cup \{j\} \in  \mwcvp,$$
	for all $S \in 2^{N \setminus\{i,j\}}$. 
	A solution $R$ satisfies the  \textit{anonymity of minimal winning coalitions}  property if
	$$i R^v j \Leftrightarrow iR^{v_{\pi}} j.$$
\end{property}

\begin{example}\label{ex:coalan}
Consider the weighted majority game $(N,v)$ of Example \ref{ex:nomon} and the players $3$ and $4$ in the role of players $i$ and $j$ of the definition of Property \ref{AMWC}.
Define a bijection $\pi$ on $2^{\{1,2,5\}}$ such that $\pi(\{1,5\})=\{2,5\}$. So the simple game $(N, v_{\pi})$ is such that
\[
\mwcvp=\{\{1,2\}, \{1,3,4\}, \{1,3,5\}, \{2,4,5\} \}.
\]
Game $v$ differs from $v_{\pi}$ in terms of minimal winning coalitions just for coalition $\{1,4,5\}$ which is replaced in $v_{\pi}$ by the minimal winning coalition $\{2,4,5\}$. Nevertheless,  the number  of minimal winning coalitions of each size containing player $4$ in game in $v_{\pi}$ is precisely as in game in $v$, so her capacity to form minimal winning coalitions should not be affected (assuming that the other players are equally inclined to form minimal winning coalitions with $4$). So, the property of Anonymity of Minimal Winning Coalitions says that the relative ranking between $3$ and $4$ in $v$ should be the same  as in $v_{\pi}$. 
\end{example}
Property \ref{AMWC} reflects a broadly adopted principle, satisfied by classical power indices like the Shapley-Shubik index \cite{Shapley54}, the Banzhaf index \cite{Banzhaf65} and all semivalues \cite{Dubey81}, saying that coalitions of the same size are equally likely. So, it seems compelling to assume that the relative position of two players is not affected by permutations preserving the size of minimal winning coalitions containing them, as it is required by Property \ref{AMWC}.

Another property we consider in our analysis is the one of independence of larger minimal winning coalitions, saying that, once a solution exists, in which a player $i$ is ranked strictly better than a player $j$, adding ``larger'' minimal winning coalitions should not affect the relative ranking between $i$ and $j$.
\begin{property}[Independence of Larger Minimal Winning Coalitions (ILMWC)]\label{IHC} \rm
	Let $i,j \in N$. For any  $v \in \mathcal{SG}^N$, let $h=\max\{|S|:S \in \mwcv \mbox{ and } S \cap \{i,j\} \neq \varnothing\}$ be the highest cardinality of coalitions in the set $ \mwcv$ containing either $i$ or $j$. Let $\mathcal{S}_h$ be a collection of (minimal) winning coalitions with cardinality strictly larger than $h$, i.e., 
	$\mathcal{S}_h=\{S_1, \ldots, S_r\}$ such that $S_k \subseteq N$, $|S_k| >h$ for $k=1, \ldots, r$ and there is no $Q \in  \mwcv\cup \mathcal{S}_h$ with $Q \subset S_k$,
	for all $k \in \{1, \ldots, r\}$.
	A solution $R$ satisfies the  \textit{independence of larger minimal winning coalitions}  property if
	$$
	i P^v j \Rightarrow iP^{v'} j,
	$$
	where $v'$ is a simple game such that the set of minimal winning coalitions  is obtained as  $ \mwcpv= \mwcv \cup \mathcal{S}_h$.
	
\end{property}

\begin{example}\label{ex:ihc}
Consider again the weighted majority game $(N,v)$ of Example \ref{ex:nomon} and the player $1$ and $2$ in the role of players $i$ and $j$ of the definition of Property \ref{IHC}.
Let $\mathcal{S}_h=\{\{2,3,4,5\}\}$ and consider a new simple game $(N, v')$ such that
\[
\mwcpv =\{\{1,2\}, \{1,3,4\}, \{1,3,5\}, \{1,4,5\}, \{2,3,4,5\} \}.
\]
Notice that the new simple game $v'$ contains one more minimal winning coalition containing $2$ but not $1$, but the size of such a minimal winning coalition in $v'$ is strictly larger than the size of any minimal winning coalition in $v$, and therefore is considered less likely to form. If a solution satisfying the property of independence of larger minimal winning coalitions ranks $1$ strictly better than $2$ in the simple game $v$,  in $v'$ the solution also must rank $1$ strictly better than $2$: the new (and larger) minimal winning coalition does not affect the strict ranking decided on the basis of smaller minimal winning coalitions.
\end{example}

In collective decision-making bodies, forming large winning coalitions in practice may result more difficult than forming small ones due to many factors, like the presence of complex institutional rules, the need of mediators in the decision-making process, higher negotiation costs or other ``psychological'' aspects, like contrasting political positions of their members. As a consequence, it is crucial 
to emphasize the impact of minimal winning coalitions of small size, as demanded by Property \ref{IHC}, which preserves strict rankings after the addition of large minimal winning coalitions.

\begin{proposition}\label{prop:DCAPR}
	Let $R$ be a solution satisfying Properties \ref{PD} (DM)  and \ref{AMWC} (AMWC). Then  for any simple game $v$ and $i,j \in N$ such that $\theta_v(i)=\theta_v(j)$ we have that $i I^v j$.
\end{proposition}
\begin{proof}
Since $\theta_v(i)=\theta_v(j)$, we have that $i_k=j_k$ for all $k \in \{1, \ldots, n\}$. Define a bijection $\pi$  on $2^{N \setminus\{i,j\}}$ such that for each $k \in \{1, \ldots, n-1\}$ and for each coalition $S \in 2^{N \setminus\{i,j\}}$  
of size $k-1$ with $S \cup \{j\} \in \mwcv$, $\pi(S)=T$, where $T \in 2^{N \setminus\{i,j\}}$ is a coalition of size $k-1$, with $T \cup \{i\} \in \mwcv$. Consider a game $v_{\pi}$ such that $ S \cup \{i\} \in \mwcv \Leftrightarrow  S \cup \{i\} \in \mwcvp$
and $ S \cup \{j\} \in \mwcv \Leftrightarrow  \pi(S) \cup \{j\} \in \mwcvp$. So, we have that for all $k \in \{1, \ldots, n\}$ and all coalitions $T \in 2^{N \setminus\{i,j\}}$ of size $k-1$ with $T \cup \{i\} \in \mwcv$
\[
T \cup \{i\} \in \mwcvp \Leftrightarrow  T \cup \{j\} \in \mwcvp.
\]
Then, $i \sim^{v_{\pi}} j$ (players $i$ and $j$ are equally desirable in $v_{\pi}$) and, by Property \ref{PD}
\begin{equation}\label{rel:equivpd}
i I^{v_{\pi}} j.
\end{equation}
 Notice that $i,j, \pi,$ $v$ and $v_{\pi}$  satisfy the conditions for bijections demanded in the statement of Property \ref{AMWC}, with $v_{\pi}$ such that the set of minimal winning coalitions of $v_{\pi}$ is 
 \[
 \mwcvp=\bigcup_{T \in 2^{N \setminus\{i,j\}} \mbox{ s.t. } T \cup \{i\} \in\mwcv} \{ T \cup \{i\}, T \cup \{j\}\}
 \] 
(notice that the minimality of the elements in $\mwcvp$ is guaranteed by the minimality of the elements in $\mwcv$).
So, since $R$ satisfies Property \ref{AMWC}, we have that
\begin{equation}\label{DCAapp}
i I^v j \Leftrightarrow iI^{v_{\pi}} j.
\end{equation}
So, by relation (\ref{rel:equivpd}), we have $i I^v j$, which concludes the proof.
\end{proof}

\begin{theorem}\label{2}
The LRM solution $R_{\ran}$ is the unique solution that fulfils Properties \ref{PD} (DM), \ref{AMWC} (AMWC) and \ref{IHC} (ILMWC).
\end{theorem}
\begin{proof} 
 By Proposition \ref{prop:pd} we have that $R_{\ran}$ fulfils Property \ref{PD} (DM). It is easy to check that it also fulfils Properties \ref{AMWC} (AMWC) and  \ref{IHC} (ILMWC) (it directly follows from Definition \ref{def:LRM} and the lexicographic relation).

To show that $R_{\ran}$ is the unique index fulfilling Properties \ref{prop:pd} (DM), \ref{AMWC} (AMWC) and  \ref{IHC} (ILMWC), we need to prove that, if a solution $R: \mathcal{SG}^N \rightarrow \mathcal{T}^N$ satisfies Properties  \ref{prop:pd} (DM), \ref{AMWC} (AMWC) and  \ref{IHC} (ILMWC), then
$ iR_{\ran}^{v} j \Leftrightarrow iR^{v} j$ or, equivalently, $ i P_{\ran}^{v} j \Leftrightarrow i P^{v} j$ and $ i  I_{\ran}^{v} j \Leftrightarrow i I^{v} j.$\\ \\
\noindent
We first prove that $ i P_{\ran}^{v} j \Leftrightarrow i P^{v} j$:\\
\noindent
($\Rightarrow$)\\
Let $iP_{\ran}^{v} j$. By Definition \ref{def:critbased}, let $k'$ be the smallest integer in $\{1, \ldots, n\}$ with $i_{k'} > j_{k'}$. Let $s=i_{k'} - j_{k'}$ and $\mathcal{S}^i_{k'}=\{S \in \mwcv: |S|=k' \mbox{ and } S \cap \{i,j\}=i\}$ be a subset of  coalitions in $\mwcv$ of size $k'$ containing $i$ but not $j$ such that $|\mathcal{S}^i_{k'}|=s$. Moreover, let $\Sigma= \{S \in \mwcv: |S|>k'\}$ be the set of coalitions in $\mwcv$ with cardinality strictly larger than $k'$. 

Consider a new simple game $v'$ such that $\mathcal{W}^{v'}_{\min}=\mwcv \setminus  \Sigma$,
and the set of minimal winning coalitions containing $j$ (of size at most $k'$) in $\mathcal{W}^{v'}_{\min}$:
\[
\mathcal{S}^j=
\{S\cup \{j\}: S \in 2^{N \setminus\{i,j\}} \mbox{ with } S \cup \{j\} \in \mathcal{W}^{v'}_{\min}\}. 
\]


Define a bijection $\pi$  on $2^{N \setminus\{i,j\}}$ such that for each $t \in \{1, \ldots, k'\}$ and for each coalition $S \in 2^{N \setminus\{i,j\}}$  
of size $t-1$ with $S \cup \{j\} \in \mathcal{W}^{v'}_{\min}$, $\pi(S)=T$, where $T \in 2^{N \setminus\{i,j\}}$ is a coalition of size $t-1$, with $T \cup \{i\} \in \mathcal{W}^{v'}_{\min}$. So, the set of minimal winning coalitions contained in $\mathcal{S}^j$ after the transformation via $\pi$ is:
\[
\mathcal{T}^j=\{\pi(S)\cup \{j\}: S \in 2^{N \setminus\{i,j\}} \mbox{ with } 
S \cup \{j\} \in \mathcal{S}^j\}.
\]
Consider a new game $\hat{v}_{\pi}$ such that 
\[
\mathcal{W}^{\hat{v}_{\pi}}_{\min}=\left( \mathcal{W}^{v'}_{\min} \setminus \mathcal{S}^j \right)\cup \mathcal{T}^j .
\]
So, we have that for all coalitions $T \in 2^{N \setminus\{i,j\}}$ of size $t-1$, $t \in \{1, \ldots, k'\}$, 
\[
T \cup \{j\} \in \mathcal{W}^{\hat{v}_{\pi}}_{\min} \Rightarrow  T \cup \{i\} \in \mathcal{W}^{\hat{v}_{\pi}}_{\min},
\]
and, consequently, for all $S \in 2^{N \setminus\{i,j\}}$,
\[
S \cup \{j\} \in \mathcal{W}^{\hat{v}_{\pi}} \Rightarrow  S \cup \{i\} \in \mathcal{W}^{\hat{v}_{\pi}},
\]
which means that $ i \succeq^{\hat{v}_{\pi}} j$. So, by Property \ref{PD} (DM), we have that
$iP^{\hat{v}_{\pi}} j$.

On the other hand, 
$$ S \cup \{i\} \in \mathcal{W}^{v'}_{\min} \Leftrightarrow S \cup \{i\} \in \mathcal{W}^{\hat{v}_{\pi}}_{\min} $$
and
$$  S \cup \{j\} \in  \mathcal{W}^{v'}_{\min}\Leftrightarrow \pi(S)\cup \{j\} \in  \mathcal{W}^{\hat{v}_{\pi}}_{\min},$$
for all $S \in 2^{N \setminus\{i,j\}}$, and therefore, by Property \ref{AMWC} on $R$ applied to $v'$ and $\hat{v}_{\pi}$, we also have $iP^{v'} j$.

Finally, by  Property \ref{IHC} on $R$ (with $v'$ in the role of $v$ in the statement of  Property \ref{IHC}), we have that $iP^{v} j$, as $\mwcv=\mathcal{W}^{v'}_{\min} \cup \Sigma$.\\
\noindent
($\Leftarrow$)\\
Let $iP^{v} j$. Suppose that $i I_{\ran}^{v} j$. Then, by Definition \ref{def:critbased}, $\theta_{v}(i)=\theta_{v}(j)$. So, by Proposition \ref{prop:DCAPR}, $i I^{v} j$, which yields a contradiction with $i P^{v} j$. Since it can't even be $j P_{\ran}^{v} i$ (by the other implication proved above), and by the fact that $P^v_{\ran}$ is a total relation, it must be $i P_{\ran}^{v} j$.\\ \\
\noindent
We now prove that $ i I_{\ran}^{v} j \Leftrightarrow i I^{v} j$:\\
\noindent
($\Rightarrow$)\\
Let $iI_{\ran}^{v} j$. Then, by Definition \ref{def:critbased}, $\theta_{v}(i)=\theta_{v}(j)$. So, by Proposition \ref{prop:DCAPR} and the fact that $R^{v}$ satisfies Properties \ref{PD} and \ref{AMWC}, $iI^{v} j$.\\
($\Leftarrow$)\\
Let $iI^{v} j$. As we have shown previously, $ iP_{\ran}^{v} j \Leftrightarrow iP^{v} j.$ So it is not possible that $iP_{\ran}^{v} j$ or $jP_{\ran}^{v} i$. Since $P_{\ran}$ is a total relation, it must be $iI_{\ran}^{v} j$, which concludes the proof.

\end{proof}
We end this section showing the logical independence of Properties \ref{PD}, \ref{AMWC} and \ref{IHC}.

\begin{example}
\label{non}
[No Property \ref{PD}] Given $i,j \in N$, consider the ranking solution $R_{DM}$ defined by
$$i\ R_{DM}^v\ j \qquad {\rm iff\;} \qquad v(\{i\}) \geq v(\{j\}).$$
This solution satisfies all the Properties but Property \ref{PD}.
\end{example}

\begin{example}
\label{nodca}
[No Property \ref{AMWC}] For any $i \in N$, let $B(i)$ the the largest player index within minimal winning coalitions containing player $i$, i.e. $$B^v(i)=\max_{S \in \mwcn : i \in S} \big(\min_{j \in S\setminus \{i\}} j\big).$$
Consider the ranking solution $R_{AMWC}^v$ such that
\begin{equation*}
	\begin{cases}
		i\ I_{AMWC}^v \ j  & {\rm if\;} i\sim^v j, \\
		i\ P_{AMWC}^v \ j  & {\rm if\;} \theta_v(i) >_L \theta_v(j),\\
		i\ P_{AMWC}^v \ j  & {\rm if\;} (i,j) \notin\, \succeq^v, (j,i) \notin\, \succeq^v, \theta_v(i) = \theta_v(j) \mbox{ and } B^v(i)>B^v(j),\\
	\end{cases} 
\end{equation*}

This solution satisfies all the  Properties but Property \ref{AMWC}. 

[It is clear that $R_{AMWC}^v$ satisfies properties \ref{PD} and \ref{IHC}. To see that $R_{AMWC}^v$ does not satisfy Property \ref{AMWC}, consider, for instance, games $v$ and $v_\pi$ of Example \ref{ex:coalan}. As we noticed, a solution satisfying Property \ref{AMWC} should rank players $3$ and $4$ in the same way in both games $v$ and $v_\pi$, However, since $B^v(3)=B^v(4)=1$ in $v$ and $B^{v_{\pi}}(3)=1$ and $B^{v_{\pi}}(4)=2$ (and the two players are not in  a desirable relation in both games) we have that $3\ I_{AMWC}^v \ 4$, while $4\ P_{AMWC}^{v_{\pi}} \ 3$.]
\end{example}

\begin{example}
\label{noihc}
[No Property \ref{IHC}] For each $i \in N$, let $\overline{\theta}_v$ be the $n$-dimensional vector $\overline{\theta}_v(i)=(i_n,\dots,i_1)$ associated to $v$. Given $i,j \in N$, consider the vector  ranking solution $R_{ILMWC}^v$ such that
\begin{equation*}
	\begin{cases}
		i\ P_{ILMWC}^v \ j  & {\rm if\;} i\succ^v j, \\ 
		i\ I_{ILMWC}^v \ j  & {\rm if\;} \theta_v(i) = \theta_v(j), \\
		i\ P_{ILMWC}^v \ j  & {\rm if\;} (i,j) \notin\, \succeq^v, (j,i) \notin\, \succeq^v \mbox{ and } \overline{\theta}_v(i) >_L \overline{\theta}_v(j),\\
	\end{cases} 
\end{equation*}

This solution satisfies all the Properties but Property \ref{IHC}.

[It is easy to verify that $R_{ILMWC}^v$ satisfies Properties \ref{PD} and \ref{AMWC}. To see that $R_{ILMWC}^v$ does not satisfy property \ref{IHC}, consider, for instance, games $v$ and $v'$ of Example \ref{ex:ihc}. Notice that $1\succ^v 2$ and, so, $1\ P_{ILMWC}^v \ 2 $. However, in game $v'$, $(1,2) \notin \succeq^{v'}$ and $(2,1) \notin \succeq^{v'}$ ($1$ and $2$ are not in desirable relation), while 
$$\overline{\theta}_{v'}(2)=(0,1,0,1,0) >_L (0,0,3,1,0)=\overline{\theta}_{v'}(1)$$ and therefore $2\ P_{ILMWC}^{v'} \ 1 $.]

\end{example}

\section{Duality}\label{sec:dual}

In this section we investigate the connections between the LRM solution and the criticality-based ranking introduced in \cite{ADFM21} to rank players in a simple game. 
 In \cite{ADFM21} a ranking over players is defined according to the power of blocking the grand coalition to be winning. 
Given a simple game $(N,v)$  a coalition $B\subseteq N$ is called \emph{blocking} coalition for $N$ if $v(N\setminus B)=0$. Let $\bcv$ be the set of all blocking coalitions  in the game $(N,v)$ and let $\mbcv$ be the set of all minimal blocking coalitions $\mbcv=\Min \bcv$.
Denote by $i^*_k$ the number of minimal blocking coalitions (for $N$) of size $k$ containing player $i$, so $i^*_k=|\{B \in \mbcv: i \in B, |B|=k\}|$ for all $k \in \{1, \ldots, n\}$. For each $i \in N$, let $\theta^*_v(i)$ be the $n$-dimensional vector $\theta^*_v(i)=(i^*_1,\dots,i^*_n)$ associated to $v$.

 The criticality-based ranking is based on the idea that the smaller is the size of a blocking coalition, the larger is the influence on the blocking power of its members; the ranking of a player in terms of blocking power is positively correlated first to the size of minimal blocking coalitions the player belongs to and second to their number.

\begin{definition}\label{def:critbased}
	The \emph{criticality-based  solution} is the function $R_{\ranc}:\mathcal{SG}^N \longrightarrow \mathcal{T}^N$ defined for any  simple game $v \in \mathcal{SG}^N$ as
	$$
	i \  R_{\ranc}^v \ j \qquad {\rm if\;} \qquad \theta^*_v(i)\;\;\ge_L \;\;\theta^*_v(j).
	$$
	Let $I_{\ranc}^v$ and $P_{\ranc}^v$ be the symmetric part and the asymmetric part of $R_{\ranc}^v$, respectively.
\end{definition}

\begin{example}
Consider the simple game of Example \ref{ex: 1} then we have that 
\[
\mbcv=\{\{1,4\}, \{1,2,3\}, \{1,3,5\}, \{2,3,4\}, \{2,3,5\} \}.
\]
Therefore,
\begin{eqnarray*}
&\theta^*_{v}(1)=(0, 1, 2, 0, 0),\ &\theta^*_{v}(2)=(0, 0, 3, 0, 0),\\
&\theta^*_{v}(3)=(0, 0, 4, 0, 0),\ &\theta^*_{v}(4)=(0, 1, 1, 0, 0),\\
&\theta^*_{v}(5)=(0, 0, 2, 0, 0).
\end{eqnarray*}
So, the criticality-based ranking is such that
\begin{equation*}
	1\ P_{\ranc}^v\ 4\ P_{\ranc}^v\ 3\ P_{\ranc}^v\ 2\ P_{\ranc}^v\ 5.
\end{equation*}

\end{example}

We first show that the LRM coincides with the criticality-based ranking of the dual game.

\begin{proposition}\label{prop:coinc}
Let $(N,v)$ be a simple game. Then $R_{\ran}^{v}=R_{\ranc}^{v^*}$.
\end{proposition}

\begin{proof}
Given the simple game $(N,v)$ its dual $v^*$ is defined by 
\begin{equation}\label{def:dual}
	v^*(S)=v(N) - v(N \setminus S),
\end{equation}
for each coalition $S \in 2^N$. 
The proposition follows recalling that $\mwcvv=\mbcv$, as proved in Proposition 3 in \cite{ADFM21}, and then $\theta_v=\theta^*_{v^*}$. 
\end{proof}


On the other hand, it is also interesting to study under which conditions the LRM  and the criticality-based ranking coincide. To this purpose, we analyse the  behaviour of the desirability relation on a simple game $v$ and its dual $v^*$.

\begin{proposition}\label{desdual}
Given a simple game $(N,v)$ and the dual game $(N,v^*)$ then, $\forall i,j\in N$, $i\neq j$
\begin{equation*}
i \succeq^{v} j\quad \iff \quad i \succeq^{v^*} j.
\end{equation*}
\end{proposition}
\begin{proof}
$\Rightarrow$ By hypothesis, for all $S \subseteq N\setminus\{i,j\}$, $S\cup\{j\}\in\wcv$ implies that $S\cup\{i\}\in\wcv$. We want to prove that for all $ S\subseteq N\setminus\{i,j\}$ such that $N\setminus\{S\cup\{j\}\}\notin\wcv$ implies that $N\setminus\{ S\cup\{i\}\}\notin\wcv$. Suppose that   $N\setminus\{ S\cup\{i\}\}$ is winning then define $T:=N\setminus\{ S\cup\{i,j\}\}$. We observe that $T\cup\{j\}$ is winning then $T\cup\{i\}=N\setminus\{ S\cup\{j\}\}$ is winning, i.e. a contradiction.\\
$\Leftarrow$ By hypothesis, $\forall S\subseteq N\setminus\{i,j\}$, $N\setminus\{S\cup\{j\}\}\notin\wcv$ implies that $N\setminus\{S\cup\{i\}\}\notin\wcv$. We want to prove that for all $ S\subseteq N\setminus\{i,j\}$, $S\cup\{j\}\in\wcv$ implies that $S\cup\{i\}\in\wcv$. Suppose $S\cup\{i\}\notin\wcv$ and let $T =N\setminus\{S\cup\{i,j\}\}$. We observe that on the one hand $N\setminus\{T\cup\{i\}\}$ is not winning, but on the other hand $N\setminus\{T\cup\{i\}\}=S\cup\{j\}$ is winning, i.e. a contradiction,   and the proof is complete.
\end{proof}

\begin{corollary}
Let $(N,v)$ be a simple game such that the desirability relation is total. Then, $R_{\ran}^v=R_{\ranc}^v$. Moreover the LRM and the criticality-based ranking are self-dual.
\end{corollary}


\begin{example}\label{ex:weightcoinc}
Consider the weighted majority game $(N,v)$ of Example \ref{ex:nomon}. The minimal winning and blocking coalitions are
\[
\mwcv =\{\{1,2\}, \{1,3,4\}, \{1,3,5\}, \{1,4,5\} \},
\]
\[
\mbcv =\{\{1\}, \{2,3,4\}, \{2,3,5\}, \{2,4,5\} \}.
\]
We have that
\begin{eqnarray*}
&\theta^*_{v}(1)=(1, 0, 0, 0, 0),\ &\theta^*_{v}(2)=(0, 0, 3, 0, 0),\\
&\theta^*_{v}(3)=(0, 0, 2, 0, 0),\ &\theta^*_{v}(4)=(0, 0, 2, 0, 0),\\
&\theta^*_{v}(5)=(0, 0, 2, 0, 0).
\end{eqnarray*}
The criticality-based ranking is such that
\begin{equation*}
	1\ P_{\ranc}^v\ 2\ P_{\ranc}^v\ 3\ I_{\ranc}^v\ 4\ I_{\ranc}^v\ 5.
\end{equation*}

So, as expected for a weighted majority game in which the desirable relation is total, $R_{\ran}^v = R_{\ranc}^v $.

\end{example}

With the purpose of ranking players in a simple game according to their influence in the process of forming blocking coalitions, it seems natural to look at a dual version of Property  \ref{PD}.
\begin{property}[Dual Desirable Monotonicity (DDM)]\label{DPD} \rm
Let $i,j \in N$. For any  $v \in \mathcal{SG}^N$, a solution $R$ satisfies the  \textit{dual desirable monotonicity}  property if
$$
i \sim^{v^*} j \Rightarrow i I^{v} j,
$$
and
$$
i \succ^{v^*} j \Rightarrow i P^{v} j.
$$
\end{property}
A solution satisfying Property \ref{DPD} obeys to the desirability relation defined on dual games, specifying that $i$ is at least as desirable as $j$ if we can replace player $j$ with player $i$ in any blocking coalition (instead of in any winning one). In a similar fashion, Properties \ref{AMWC} and \ref{IHC} can be reformulated as their following dual counterparts. 
\begin{property}[Dual Anonymity of Minimal Winning Coalitions (DAMWC)]\label{DAMWC} \rm
	Let $i,j \in N$,  $v, v_{\pi} \in \mathcal{SG}^N$ and let $\pi$ be a bijection on $2^{N \setminus\{i,j\}}$ with $|\pi(S)|=|S|$ and such that
	$$ S \cup \{i\} \in \mwcvv \Leftrightarrow  S \cup \{i\} \in \mwcvvp$$
	and
	$$ S \cup \{j\} \in  \mwcvv \Leftrightarrow  \pi(S) \cup \{j\} \in  \mwcvvp,$$
	for all $S \in 2^{N \setminus\{i,j\}}$. 
	A solution $R$ satisfies the  \textit{anonymity of minimal winning coalitions}  property if
	$$i R^v j \Leftrightarrow iR^{v_{\pi}} j.$$
\end{property}
\begin{property}[Independence of Larger Minimal Winning Coalitions in the Dual (ILMWCD)]\label{IHCD} \rm
	Let $i,j \in N$. For any  $v \in \mathcal{SG}^N$, let $h=\max\{|S|:S \in \mwcv \mbox{ and } S \cap \{i,j\} \neq \varnothing\}$ be the highest cardinality of coalitions in the set $ \mwcvv$ containing either $i$ or $j$. Let $\mathcal{S}_h$ be a collection of (minimal) winning coalitions in the dual game $v^*$ with cardinality strictly larger than $h$, i.e., 
	$\mathcal{S}_h=\{S_1, \ldots, S_r\}$ such that $S_k \subseteq N$, $|S_k| >h$ for $k=1, \ldots, r$ and there is no $Q \in  \mwcvv\cup \mathcal{S}_h$ with $Q \subset S_k$,
	for all $k \in \{1, \ldots, r\}$.
	A solution $R$ satisfies the  property of \textit{independence of larger minimal winning coalitions in the dual}  if
	$$
	i P^v j \Rightarrow iP^{v'} j,
	$$
	where $v'$ is a simple game such that the set of minimal winning coalitions  is obtained as  $ \mwcpvv= \mwcvv \cup \mathcal{S}_h$.
\end{property}
We can state the following result.
\begin{theorem}\label{dual2}
The solution $R_{\ran'}$ such that $R_{\ran'}^{v}=R_{\ran}^{v^*}$ for all $v \in \mwcv$ is the unique solution that fulfils Properties \ref{DPD}, \ref{DAMWC} and \ref{IHCD}.
\end{theorem}
\begin{proof}
The proof follows the same steps of the proof of Theorem \ref{2}, with $v^*$ in the role of $v$.  
\end{proof}
By Proposition \ref{prop:coinc}, and the fact that $(v^*)^*=v$ (the dual of the dual of a game $v$ equals game $v$), we have that $R_{\ran'}=R_{\ranc}$. Moreover, by Proposition \ref{desdual}, we have that  Properties \ref{PD} and \ref{DPD} are equivalent. So, the following corollary holds.
\begin{corollary}\label{dual4}
The critcality-based solution $R_{\ranc}$ is the unique solution that fulfils Properties \ref{PD}, \ref{DAMWC} and \ref{IHCD}.
\end{corollary}
In \cite{ADFM21}, the critcality-based solution has been axiomatically characterized using four properties, namely, \textit{Players' Anonymity}, \textit{Dual Coalitional Anonymity}, \textit{Dual Monotonicity} and \textit{Independence of Higher Cardinalities}(see Sections 4 in \cite{ADFM21} for a formal definitions of these axioms). Notice that Property \ref{DAMWC} coincides with the property of \textit{Dual Coalitional Anonymity} in \cite{ADFM21}, while Property \ref{IHCD} coincides with the property of Independence of Higher Cardinalities in \cite{ADFM21}. So, according to Corollary \ref{dual4}, Property \ref{PD} replaces properties of Players' Anonymity and  Dual Monotonicity in the axiomatic characterization of the criticality-based solution presented in \cite{ADFM21}.

\begin{example}\label{ex:Freixas}
	Consider the simple game $(N,v)$ in Example 2.7 in \cite{Freixas10}:
	\[
	 \mwcv=\{\{1,2,3\},\{1,2,4\},\{1,2,5\},\{1,3,4\},\{3,4,5\}\};
	\]
	\[
	\mbcv=\{\{1,3\},\{1,4\},\{1,5\},\{2,3\},\{2,4\},\{3,4,5\}\};
	\]
The LRM solution in $v$ is 
\begin{equation*}
	1\ P_{\ell}^v\ 2\ I_{\ell}^v\ 3\ I_{\ell}^v \ 4\ P_{\ell}^v\ 5.
\end{equation*}
and the LRM solution in $v^*$ is
\begin{equation*}
	1\ P_{\ell}^{v^*}\ 3\ I_{\ell}^{v^*} \ 4\ P_{\ell}^{v^*}\ 2\ P_{\ell}^{v^*}\ 5.
\end{equation*}
Then if the desirability relation is not a total preorder the LRM solution in not self-dual. The same result holds for the criticality-based solution. 
\end{example}
We conclude this section pointing out that the  axioms of Dual Coalitional Monotonicity  and of Players' Anonymity are replaced by Desirable Monotonicity (Property \ref{PD}) in the characterization of the criticality-based solution in \cite{ADFM21}, as shown by Corollary \ref{dual4}, but the two axioms do not imply Property \ref{PD}. In fact, as shown by Example \ref{ex:nomon}, the ranking over players represented by the PGI does not satisfy Property \ref{PD}  (we leave to the reader to check that such a ranking satisfy both Dual Coalitional Monotonicity  and Players' Anonymity axioms).

\section{Conclusions}\label{sec:concl}
In this paper, using the minimal winning coalitions of a simple game, we introduced a new ranking among players that satisfies the desirability relation: the Lexicographic Ranking based on Minimal winning coalitions. The players are ranked according to the size of the minimal winning coalitions they belong to and then to the number of such coalitions.  The ranking solution satisfies the coalitional anonymity property and the independence of larger minimal winning coalitions property that together with a monotonicity property rooted on the desirability  relation uniquely characterized it. Looking at the dual game, we prove that there is a relation between the Lexicographic Ranking based on Minimal winning coalitions and the criticality-based ranking and, consequently, between ranking players according to their power to win and to their power to block the grand coalition. In particular, if the desirable relation is total the two rankings coincide. \\
Following this line of research, it would be interesting to delve more into the connection between the power to initiate and the power to block a winning coalition \cite{Deegan80}, in particular, when the desirability relation between two players does not hold.

\paragraph*{Acknowledgment} The authors want to thank professor Marco Dall'Aglio for the valuable discussions and the anonymous referee for the useful comments. \\
S. Moretti gratefully acknowledges the support of the ANR project THEMIS (ANR-20-CE23-0018)


\begin{thebibliography}{99}


\bibitem{ADFM21}
Aleandri, M., Dall'Aglio, M., Fragnelli, V., Moretti, S. (2021) Minimal winning coalitions and orders of criticality. \textit{Annals of Operations Research}. https://doi.org/10.1007/s10479-021-04199-6

\bibitem{Freixas10} 
Alonso-Meijide, J.M., Freixas, J. (2010) A new power index based on minimal winning coalitions without any surplus. \textit{Decision Support Systems}, 49(1), 70-76.

\bibitem{Banzhaf65}
Banzhaf, J. (1965) Weighted voting doesn't work: a mathematical analysis.
\textit{Rutgers Law Review}, 19, 317-343.

\bibitem{Carreras95}
Carreras, F., Freixas, J. (1995) Complete simple games. \textit{Mathematical Social
Sciences}, 32, 139-155.

%


\bibitem{DeeganPackel80}
Deegan, J., Packel, E.W. (1980) An axiomated family of power indices for simple n-person games. \textit{Public Choice}, 35, 229-239.

\bibitem{Deegan80}
Deegan, J., Packel, E.W. (1980) A new index of power for simple n-person
games. \textit{International Journal of Game Theory}, 7, 113-123.

\bibitem{Diffo02}
Diffo Lambo, L., Moulen, J.(2002) Ordinal equivalence of power
notions in voting games. \textit{Theory and Decision}, 53, 313-325.

\bibitem{Dubey81}
Dubey, P., Neyman, A.,  Weber, R. (1981). Value theory without efficiency. \textit{Mathematics of Operations
Research}, 6(1), 122-128.

\bibitem{Ferto20}
Fert\"o, I., K\'oczy, L.\'A., Kov\'acs, A., Sziklai, B.R. (2020) The power ranking of the members of the Agricultural Committee of the European Parliament. \textit{European Review of Agricultural Economics}, 47(5), 1897-1919.

\bibitem{Freixas97}
Freixas, J., Gambarelli, G. (1997) Common internal properties among power indices. \textit{Control and Cybernetics}, 26, 591-604.

\bibitem{Freixas12}
Freixas, J., Marciniak, D.,  Pons, M. (2012) On the ordinal equivalence of the Johnston, Banzhaf and Shapley power indices. \textit{European Journal of Operational Research}, 216(2), 367-375.

\bibitem{Freixas05}
Freixas, J., Pons, M. (2005). Two measures of circumstantial power: Influences and bribes. \textit{Homo Oeconomicus}, 22, 569-588

\bibitem{Freixas08}
Freixas, J.,  Pons, M. (2008) Circumstantial power: Optimal persuadable voters. \textit{European
Journal of Operational Research}, 186(3), 1114-1126.



\bibitem{H82}
 Holler, M.J. (1982) Forming coalitions and measuring voting power. \textit{Political Studies}, 30, 262-271.
%

\bibitem{Holler83}
Holler, M.J., Packel, E.W. (1983) Power, luck and the right index. \textit{Journal of Economics}, 43,
21-29.

\bibitem{Holler13}
Holler, M.J.,  Nurmi, H. (Eds.). (2013) Power, voting, and voting power: 30 years after. Springer Berlin Heidelberg.

\bibitem{Isbell58}
Isbell, J.R. (1958) A class of simple games. \textit{Duke Mathematical Journal},
25, 423-439.





\bibitem{Johnston78}
Johnston, R.J. (1978) On the measurement of power: some reactions to
Laver. \textit{ Environment and Planning A}, l0(8), 907-914.

\bibitem{Lindelauf13}
Lindelauf, R.H., Hamers, H.J.,  Husslage, B.G.M. (2013). Cooperative game theoretic centrality analysis of terrorist networks: The cases of jemaah islamiyah and al qaeda. \textit{European Journal of Operational Research}, 229(1), 230-238.

\bibitem{Moretti07}
Moretti, S., Patrone, F.,  Bonassi, S. (2007). The class of microarray games and the relevance index for genes. \textit{Top}, 15(2), 256-280.

\bibitem{Schmeidler69}
Schmeidler, D. (1969) The nucleolus of a characteristic function game. \textit{SIAM
Journal of Applied Mathematics}, 17, 1163-1170.


\bibitem{Shapley54}
Shapley, L.S., Shubik, M. (1954) A method for evaluating the distribution
of power in a committee system. \textit{American Political Science Review},
48, 787-792.

\bibitem{Zwicker99}
Taylor, A.D., Zwicker, W.S. (1999) Simple games: Desirability relations, trading, pseudoweightings. Princeton University Press.



\end{thebibliography}
\end{document}